\documentclass[11pt]{article}

\usepackage[margin=1in]{geometry}  
\usepackage{graphicx}              
\usepackage{amsmath}               
\usepackage{amsfonts}              
\usepackage{amsthm}                
\usepackage{float}				   
\usepackage{hyperref}
\usepackage{color}									
\usepackage{centernot}
\usepackage[usenames,dvipsnames]{xcolor}

\newtheorem{lemma}{Lemma}[section]
\newtheorem{theorem}{Theorem}[section]
\newtheorem{proposition}{Proposition}[section]

\newtheorem{corollary}[lemma]{Corollary}

\newcommand{\Z}{\mathbb{Z}}

\newcommand{\R}{\mathbb{R}}
\newcommand{\C}{\mathbb{C}}

\newcommand{\leqslant}{\le}

\newcommand{\real}{\operatorname{Re}}

\newtheorem{definition}{Definition}[section]
\numberwithin{equation}{section}		 			
\numberwithin{figure}{section}			 			

\allowdisplaybreaks[4]

\begin{document}
\title{Zeta functions on tori using contour integration}
\author{Emilio Elizalde, Klaus Kirsten, Nicolas Robles and Floyd Williams}
\date{}
\maketitle
\begin{abstract}
\noindent A new, seemingly useful presentation of zeta functions on complex tori is derived by using contour integration. It is shown to agree with the one obtained by using the Chowla-Selberg series formula, for which an alternative proof is thereby given. In addition, a new proof of the functional determinant on the torus results, which does not use the Kronecker first limit formula nor the functional equation of the non-holomorphic Eisenstein series. As a bonus, several identities involving the Dedekind eta function are obtained as well.
\end{abstract}
\tableofcontents
\section{Introduction}
Zeta regularization and the theory of spectral zeta functions are powerful and elegant techniques that allow one to assign finite values to otherwise manifestly infinite quantities in a unique and well-defined way \cite{Elizalde,EORBZ,kirs02b}.\\
Suppose we have a compact smooth manifold $M$ with a Riemannian metric $g$ and a corresponding Laplace-Beltrami operator $\Delta=\Delta(g)$, where $\Delta$ has a discrete spectrum
\begin{align}
0 = \lambda _0  < \lambda _1  < \lambda _2  <  \cdots ,\quad \mathop {\lim }\limits_{j \to \infty } \lambda _j  = \infty .
\end{align}
If we denote by $n_j$ the finite multiplicity of the $j$-th eigenvalue $\lambda_j$ of $\Delta$ then, by a result of H. Weyl \cite{Weyl}, which says that the asymptotic behavior of the eigenvalues as $j \to \infty$ is $\lambda _j  \sim j^{2/\dim M}$, we can construct the corresponding spectral zeta function as
\begin{align}
\zeta _M (s) = \sum\limits_{j = 1}^\infty  {\frac{{n_j }}{{\lambda _j^s }}},
\end{align}
which is well-defined for $\operatorname{Re} (s) > \tfrac{1}{2}\dim M$. Minakshisundaram and Pleijel \cite{MP} showed that $\zeta _M (s)$ admits a meromorphic continuation to the whole complex plane and that, in particular, $\zeta _M (s)$ is holomorphic at $s=0$. This, in turn, means that $\exp [ - \zeta _M '(0)]$ is well-defined and Ray and Singer \cite{RaySinger} set the definition
\begin{align} \label{detdef}
\det (\Delta)  = \prod\limits_{k = 1}^\infty  {\lambda _k^{n_k } } := e^{ - \zeta _M '(0)} ,
\end{align}
where it is understood that the zero eigenvalue of $\Delta$ is not taken into the product. For the reader's understanding, the motivation for this definition comes from the formal computation
\begin{align}
\exp \left[ { - \left. {\frac{d}{{ds}}} \right|_{s = 0} \sum\limits_{k = 1}^\infty  {\frac{{n_k }}{{\lambda _k^s }}} } \right] = \exp \left[ {\sum\limits_{k = 1}^\infty  {n_k \log \lambda _k } } \right] = \prod\limits_{k = 1}^\infty  {e^{n_k \log \lambda _k } }  = \prod\limits_{k = 1}^\infty  {\lambda _k^{n_k } }.
\end{align}
As long as the spectrum is discrete, definition \eqref{detdef} is suitable for more general operators on other infinite dimensional spaces. In particular, it is useful for Laplace-type operators on smooth manifolds of a vector bundle over $M$, see e.g.
\cite{Nakahara,Williams1}.
\section{Argument principle technique}
Let us introduce the basic ideas used in this article by considering a
generic one dimensional second order differential operator $\mathcal{O}:=  - d^2 /dx^2 + V(x)$ on the interval $[0,1]$, where $V(x)$ is a smooth potential. Let its eigenvalue problem be given by
\begin{align}
\mathcal{O}\phi _n (x) = \lambda _n \phi _n (x),
\end{align}
and choose Dirichlet boundary conditions $\phi _n (0)=\phi _n (1)=0$. This problem can be translated to a unique initial value problem \cite{Kirsten3,Kirsten4c,Levitan}
\begin{align}
(\mathcal{O} - \lambda  )u_\lambda  (x) = 0,
\end{align}
with $u_\lambda(0)=0$ and $u_\lambda'(0)=1$. The eigenvalues $\lambda_n$ then follow as the solutions to the equation
\begin{align}
u_\lambda  (1) = 0,
\end{align}
where $u_\lambda  (1)$ is an analytic function of $\lambda$. Let us recall the argument principle from complex analysis.
It states that if $f$ is a meromorphic function inside and on some counterclockwise contour $\gamma$ with $f$ having neither zeroes nor poles on $\gamma$, then
\begin{align}
\int_\gamma  dz\,\,{\frac{{f'(z)}}{{f(z)}}}  = 2\pi i(N - P),
\end{align}
where $N$ and $P$ are, respectively, the number of zeros and poles of $f$ inside the contour $\gamma$. There is a slightly stronger version of this statement called the generalized argument principle, stating that if $f$ is meromorphic in a simply connected set $D$ which has zeroes $a_j$ and poles $b_k$, if $g$ is an analytic function in $D$, and if we let $\gamma$ be a closed curve in $D$ avoiding $a_j$ and $b_k$, then
\begin{align}
\sum\limits_j {g(a_j )n(\gamma ,a_j )}  - \sum\limits_k {g(b_k )n(\gamma ,b_k )}  = \frac{1}{{2\pi i}}\int_\gamma  dz\,\,{g(z)\frac{{f'(z)}}{{f(z)}}} ,
\end{align}
where $n(\gamma ,a)$ is the winding number of the closed curve $\gamma$ with respect to the point $a \notin \gamma$, defined as
\begin{align}
n(\gamma ,a) = \frac{1}{{2\pi i}}\int_\gamma  {\frac{{dz}}{{z - a}}} \ \  \in \mathbb{Z}.
\end{align}
If we let $f$ be a polynomial with zeros $z_1, z_2, \cdots$ and $g(z):=z^{s}$, then
\begin{align}
\frac{1}{{2\pi i}}\int_\gamma  dz\,\,{z^s \frac{{f'(z)}}{{f(z)}}}  = z_1^s  + z_2^s  +  \cdots ,
\end{align}
or equivalently
\begin{align}
\frac{1}{{2\pi i}}\int_\gamma  dz\,\,{z^s \frac{d}{{dz}}\log f(z)}  = \sum\limits_n {z_n^s }.
\end{align}
Taking into account the asymptotic properties of $u_\lambda (1)$ and making the substitutions $z \to \lambda$ and $s \to -s$, we see that \cite{kirs02b,Kirsten3,Kirsten4c}
\begin{align}
\frac{1}
{{2\pi i}}\int_\gamma  {d\lambda \lambda ^{ - s} \frac{d}
{{d\lambda }}\log u_\lambda  (1)}  = \sum\limits_n {\lambda _n^{ - s} } : = \zeta _\mathcal{O} (s),
\end{align}
since the eigenvalues $\lambda_n$ are solutions of $u_\lambda(1)=0$. As before, $\gamma$ is a counterclockwise contour that encloses all eigenvalues, which we assume to be positive; see Fig.~ 2.1.
\begin{figure}[H]
	\centering
		\includegraphics{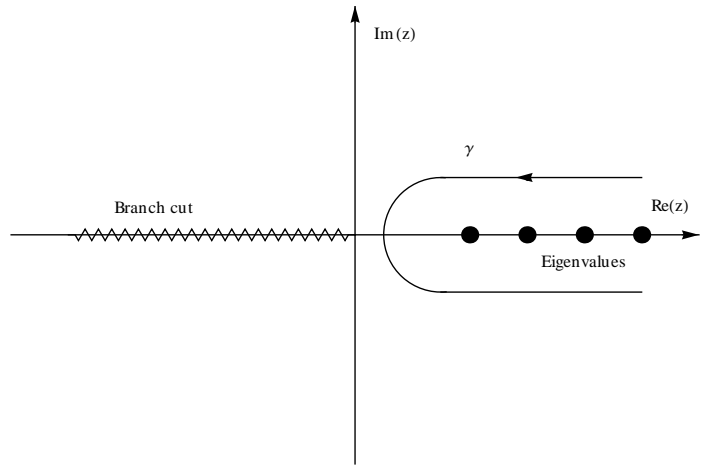}
	\caption{Integration contour $\gamma$.}
\end{figure}
The pertinent remarks for the case when finitely many eigenvalues are non-positive are given in \cite{Kirsten2}. It is important to note that the asymptotic behavior of $u_\lambda(1)$ as $\left| \lambda  \right| \to \infty$ is given by \cite{Kirsten4c,Levitan}
\begin{align}
u_\lambda  (1) \sim \frac{{\sin \sqrt \lambda  }}{{\sqrt \lambda  }}.
\end{align}
This implies that the integral representation for $\zeta _\mathcal{O}(s)$ is valid for $\operatorname{Re} (s) > \tfrac{1}{2}$ and, therefore, we must continue it analytically if we are to take its derivative at $s=0$ to compute the determinant of $\mathcal{O}$.\\\\
The next step \cite{Kirsten3,Kirsten4c} necessary to evaluate this integral is to deform the contour suitably. These deformations are allowed provided one does not cross over poles or branch cuts of the integrand. By assumption, for our integrand the poles are on the real axis and, as customary, we define the branch cut of $\lambda^{-s}$ to be on the negative real axis. This means that, as long as the behavior at infinity is appropriate, we are allowed to deform the contour to the one given in Fig.~2.2.
\begin{figure}[H]
	\centering
		\includegraphics{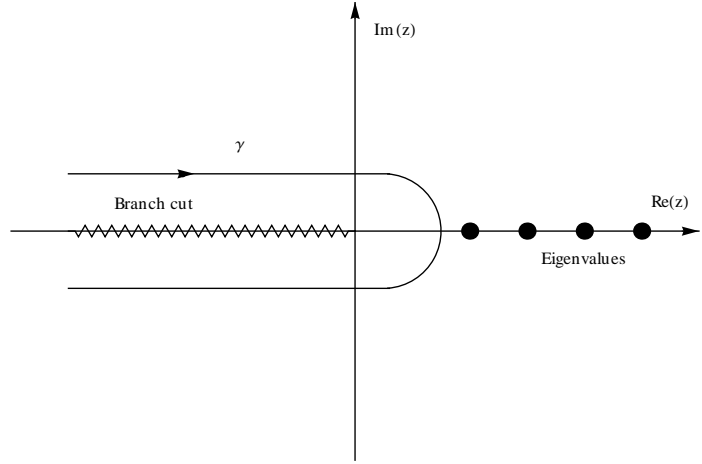}
	\caption{Deformed integration contour $\gamma$.}
\end{figure}
The result of this deformation, after shrinking the contour to the negative real axis, is
\begin{align} \label{integralopO}
\zeta _\mathcal{O} (s) = \frac{{\sin \pi s}}{\pi}\int_0^\infty  {d\lambda \lambda ^{ - s} \frac{d}{{d\lambda}}\log u_{ - \lambda } (1)}.
\end{align}
To establish the limits of the validity of this integral representation we must examine the behavior of the integrand for $\lambda \to \infty$, namely one has \cite{Kirsten2,Kirsten4c}
\begin{align}
u_{ - \lambda } (1) \sim \frac{{\sin (i\sqrt \lambda  )}}{{(i\sqrt \lambda  )}} \sim \frac{e^{\sqrt \lambda  } }{{2\sqrt \lambda  }}.
\end{align}
Thus, to leading order in $\lambda$, the integrand behaves as $\lambda ^{ - s - \tfrac{1}{2}}$, which means that convergence at infinity is established for $\operatorname{Re} (s) > \tfrac{1}{2}$, as we discussed above. On the other hand, when $\lambda \to 0$ the behavior $\lambda^{-s}$ follows. Consequently the integral representation \eqref{integralopO} is well defined for
\begin{align}
\frac{1}{2} < \operatorname{Re} (s) < 1.
\end{align}
The analytic continuation to the left is accomplished by subtracting the leading $\lambda \to \infty$ asymptotic behavior of $u_{ - \lambda } (1)$ \cite{kirs02b}.
Carrying out this procedure results in one part that is finite at $s=0$ and another part for which the analytic continuation can be constructed relatively easily.\\
The partition of $\zeta _\mathcal{O}$ that keeps the $\lambda \to 0$ term unaffected and that should improve the $\lambda \to \infty$ part is accomplished by splitting the integration range as
\begin{align}
\zeta _\mathcal{O} (s) = \zeta _{\mathcal{O},\operatorname{asy} } (s) + \zeta _{\mathcal{O},\operatorname{f} } (s),
\end{align}
where
\begin{align}
\zeta _{\mathcal{O},\operatorname{f} } (s) = \frac{{\sin \pi s}}
{\pi}\int_0^1 {d\lambda \lambda ^{ - s} \frac{d}
{{d\lambda}}\log u_{ - \lambda } (1)}  + \frac{{\sin \pi s}}
{\pi}\int_1^\infty  {d\lambda \lambda ^{ - s} \frac{d}
{{d\lambda}}\log \left[u_{-\lambda} (1) \left( \frac {2\sqrt \lambda} {e^{\sqrt \lambda}}\right)\right]} ,
\end{align}
and
\begin{align}
\zeta _{\mathcal{O},\operatorname{asy} } (s) = \frac{{\sin \pi s}}
{\pi}\int_1^\infty  {d\lambda \lambda ^{ - s} \frac{d}
{{d\lambda}}\log \frac{{e^{\sqrt \lambda  } }}
{{2\sqrt \lambda  }}} .
\end{align}
Clearly, we have constructed $\zeta _{\mathcal{O},\operatorname{f} } (s)$ in such a way that it is already analytic at $s=0$ and, thus, its derivative at $s=0$ can be computed immediately
\begin{align}
\zeta _{\mathcal{O},\operatorname{f} } '(0) =  - \log [2e^{ - 1} u_0 (1)].
\end{align}
For this case, the analytic continuation to a meromorphic function on the complex plane now follows from
\begin{align} \label{trick}
\int_1^\infty  {d\lambda \lambda ^{ - \alpha } }  = \frac{1}{{\alpha  - 1}},\quad {\text{for}}\quad \operatorname{Re} (\alpha ) > 1.
\end{align}
Applying the above to $\zeta _{\mathcal{O},\operatorname{asy} } (s)$ yields
\begin{align}
\zeta _{\mathcal{O},\operatorname{asy} } (s) = \frac{{\sin \pi s}}{{2\pi}}\left( {\frac{1}{{s - 1/2}} - \frac{1}{s}} \right),
\end{align}
and thus
\begin{align}
\zeta _{\mathcal{O},\operatorname{asy} } '(0) =  - 1.
\end{align}
The contribution from both terms then becomes
\begin{align}
\zeta _\mathcal{O} '(0) =  - \log [2u_0 (1)].
\end{align}
Note how we could numerically evaluate the determinant of $\mathcal{O}$ without using a single eigenvalue explicitly \cite{Kirsten4c}.
\section{Description of the problem}
Let us next introduce the notions needed for the investigations of the Eisenstein series.
Let $M$ be a compact smooth manifold with dimension $d$ and let $s \in \mathbb{C}$ with $\operatorname{Re} (s) > d/2$,
furthermore let $\mathbb{H}$ denote the upper half-plane $\mathbb{H} = \{\tau = \tau_1 + i \tau_2, \tau_1 \in \R, \tau_2 >0 \}$.
\begin{definition}
For $c \in \R_+$ and $\vec r \in \R_ + ^d$ the homogeneous Epstein zeta function is defined as \cite{epstein1,epstein2}
\begin{align}  \label{homogeneousEpstein}
\zeta _\mathcal{E} (s,c|\vec r): = \sum\limits_{\vec m \in \Z^d } {\frac{1}
{{(c + r_1 m_1^2  +  \cdots  + r_d m_d^2 )^s }}}.
\end{align}
If $c=0$ then it is understood that the summation ranges over $\vec m \ne \vec 0$.
\end{definition}
\begin{definition}
Let $\Z_*^2 = \Z \times \Z \setminus \{(0,0)\}$ and $\tau \in \mathbb{H}$ with $\operatorname{Re} \tau  = \tau _1$ and $\operatorname{Im} \tau  = \tau _2$ . For $\operatorname{Re} (s) > 1$ the nonholomorphic Eisenstein series is defined as \cite{Williams1}
\begin{align}
E^* (s,\tau ) := \sum\limits_{(m,n) \in \Z_*^2 } {\frac{{\tau _2^s }}{{\left| {m + n\tau } \right|^{2s} }}}.
\end{align}
\end{definition}
\noindent Note that for $\tau = i$ the nonholomorphic Eisenstein series is related to the homogeneous Epstein zeta function by
\begin{align}
E^* (s,i) = \sum\limits_{(m,n) \in \Z_*^2 } {\frac{1}{{(m^2  + n^2 )^s }}}  = \zeta _\mathcal{E} (s,0|\vec 1_2 ),
\end{align}
where $\vec 1_2  = (1,1)$. The non-holomorphic Eisenstein series is not holomorphic in $\tau$ but it can be continued analytically beyond $\operatorname{Re}(s)>1$
except at $s=1$, where there is a simple pole with residue equal to $\pi$.
\begin{definition}
\noindent For $\tau \in \mathbb{H}$, the Dedekind eta function is defined as
\[\eta (\tau ) :=  {e^{\pi i\tau /12}}\prod\limits_{n = 1}^\infty  {(1 - {e^{2\pi in\tau }})}. \]
\end{definition}
\noindent While $\eta(\tau)$ is holomorphic on the upper half-plane, it cannot be continued analytically beyond it. The fundamental properties are that it satisfies the following functional equations.
\begin{proposition}
One has
\[\begin{gathered}
  \eta (\tau  + 1) = {e^{\pi i/12}}\eta (\tau ), \hfill \\
  \eta ( - {\tau ^{ - 1}}) = \sqrt { - i\tau } \eta (\tau ), \hfill \\
\end{gathered} \]
for $\tau \in \mathbb{H}$.
\end{proposition}
\noindent The first equation is very easy to show and the proof of the second one can be found in a book of modular forms, see for instance \cite{Bump}. With this in mind, the constant term in the Laurent expansion of $E^*(s,\tau)$ is given in the following theorem (see, e.g., \cite{Williams1}).
\begin{theorem}[Kronecker's first limit formula]
\[\mathop {\lim }\limits_{s \to 1} \left( {{E^*}(s,\tau ) - \frac{\pi }{{s - 1}}} \right) = 2\pi (\gamma  - \log 2 - \log \tau _2^{1/2}\left| {\eta (\tau )} \right|^2).\]
\end{theorem}
\noindent A modular form of weight $k>0$ and multiplier condition $C$ for the group of substitutions generated by $\tau \to \tau+\lambda$ and $\tau \to - \tfrac{1}{\tau}$ is a holomorphic function $f(\tau)$ on $\mathbb{H}$ satisfying \cite{gelbart}
\begin{enumerate}
\item[(i)] $f(\tau + \lambda) = f(\tau)$,
\item[(ii)] $f(-\tfrac{1}{\tau}) = C (\tfrac{\tau}{i})^k f(\tau)$,
\item[(iii)] $f(\tau)$ has a Taylor expansion in $e^{(2\pi i \tau / \lambda)}$ (cf (i)): $f(\tau ) = \sum\nolimits_{n = 0}^\infty  {{a_n}{e^{2\pi in\tau /\lambda }}}$, i.e. ''$f$ is holomorphic at $\infty$''.
\end{enumerate}
The space of such $f$ is denoted by $M(\lambda,k,C)$ and furthermore if $a_0 = 0$ then $f$ is a cusp form. The group of substitutions generated by $\tau \to \tau+1$ and $\tau \to -\tfrac{1}{\tau}$ is
\[\operatorname{SL} (2,\Z) = \left\{ {\left( {\begin{array}{*{20}{c}}
  a&b \\
  c&d
\end{array}} \right){\text{ such that }}a,b,c,d \in \Z{\text{ and }}ad - bc = 1} \right\},\]
therefore modular forms of weight $k$ satisfy
\[f\left( {\frac{{a\tau  + b}}{{c\tau  + d}}} \right) = {(c\tau  + d)^k}f(\tau ).\]
The Dedekind eta function is a modular form of weight $k=\tfrac{1}{2}$ and we may assume without loss of generality that either $c>0$ or $c=0$ and $d=1$. Moreover, if $c=0$ and $d=1$, then it satifies
\[\eta \left( {\frac{{a\tau  + b}}{{c\tau  + d}}} \right) = \varepsilon (a,b,c,d){(c\tau  + d)^{1/2}}\eta (\tau ),\]
where $\varepsilon (a,b,c,d) = {e^{b\pi i/12}}$, and if $c>0$ then
\[\varepsilon (a,b,c,d) = \exp \left( {i\pi \frac{{a + d}}{{12}} - s(d,c) - \frac{1}{4}} \right),\]
where $s(h,k)$ is the Dedekind sum
\[s(h,k) = \sum\limits_{n = 1}^{k - 1} {\frac{n}{k}\left( {\frac{{hn}}{k} - \left\lfloor {\frac{{hn}}{k}} \right\rfloor  - \frac{1}{2}} \right)}. \]
Finally, the non-holomorphic Eisenstein series can alternatively be defined as
\[{E^*}(s,\tau ) = \zeta_R(2s)\sum_{\substack{
  (m,n) \\
  \gcd (m,n) = 1}}
  {\frac{{\tau _2^s}}{{{{\left| {m\tau  + n} \right|}^{2s}}}}}, \]
where $\zeta_R$(s) is the Riemann zeta function, and it is unchanged by the substitutions
\[\tau  \to \frac{{a\tau  + b}}{{c\tau  + d}}\]
coming from any matrix of $\operatorname{SL} (2,\Z)$. Selberg and Chowla were interested in the problem of the analytic continuation of ${E^*}(s,\tau)$ as a function of $s$ and its functional equation. Their idea was to consider the Fourier expansion of $E^*(s,\tau)$ given by
\[{E^*}(s,\tau ) = E(s,{\tau _1} + i{\tau _2}) = \sum\limits_{m \in \Z} {{a_m}(s,{\tau _2}){e^{2\pi im{\tau _1}}}} ,\]
where $a_m$ is the Fourier coefficient
\[{a_m}(\tau_2,s) = \int_0^1 {E(s,{\tau _1} + i{\tau _2}){e^{ - 2\pi im{\tau _1}}}d{\tau _1}}. \]
The explicit formulas for these coefficients are given by \cite{Bump, gelbart}
\[{a_0} = 2\zeta_R(2s)\tau _2^s + 2\phi (s)\zeta_R(2s-1)\tau _2^{1 - s},\]
where
\[\phi (s) = \sqrt \pi  \frac{{\Gamma (s - \tfrac{1}{2})}}{{\Gamma (s)}},\]
and
\[{a_n} = 2\frac{{\tau _2^{1/2}{K_{s - 1/2}}(2\pi \left| n \right|{\tau _2})}}{{{\pi ^{ - s}}\Gamma (s)}}{\left| n \right|^{s - 1}}{\sigma _{1 - 2s}}(n),\]
where for $n \ge 1$ and $v \in \C$, we let
\[{\sigma _v}(n) :=  \sum\limits_{0 < d,\;d|n} {{d^v}} \]
denote the divisor function. In general, $\phi(s)$ is called the constant or scattering term and $a_n$ with $n \ge 1$ is the non-trivial term.\\\\
In this paper, we propose to recover these results by contour integration on the complex plane without Fourier techniques.\\\\
Indeed, the existence of two (or more) general methods of obtaining analytic continuations and functional equations of zeta functions has been known since Riemann's 1859 paper on the distribution of prime numbers. The table below summarizes the dates of some of the most common zeta functions.
\begin{center} 
    \begin{tabular}{ | l | l | l | p{5cm} |}
    \hline
    $\zeta$ \textbf{function} & \textbf{Poisson summation} & \textbf{Contour integration} \\ \hline
    Riemann zeta function & Riemann (1859) & Riemann (1859) \\ \hline
		Lerch zeta function	 & Apostol (1951) & Lerch (1874) \\ \hline
    Hurwitz zeta function & Fine (1951) & Hurwitz (1882) \\ \hline
		Dirichlet $L$-function & de la Vall\'{e}e Poussin (1896) & Berndt (1973) \\ \hline
		Dedekind zeta function & Hecke (1917) &  \\ \hline
		Epstein zeta function & Epstein (1903), Chowla-Selberg (1947) & E. K. R. W. (2013) \\
    \hline
    \end{tabular}
\end{center}
\section{Eigenvalue problem set up}
For fixed $\tau=\tau_1 + i\tau_2 \in \mathbb{H}$ the general $\tau$-Laplacian is
\begin{align}
\Delta _\tau   =  - \frac{1}
{{\tau _2^2 }}\left[ {\left( {\frac{\partial }
{{\partial x}} + \tau _1 \frac{\partial }
{{\partial y}}} \right)^2  + \left( {\tau _2 \frac{\partial }
{{\partial y}}} \right)^2 } \right] =  - \frac{1}
{{\tau _2^2 }}\left[ {\frac{{\partial ^2 }}
{{\partial x^2 }} + (\tau _1^2  + \tau _2^2 )\frac{{\partial ^2 }}
{{\partial y^2 }} + 2\tau _1 \frac{\partial }
{{\partial x}}\frac{\partial }
{{\partial y}}} \right].
\end{align}
Now we let $M = S^1 \times S^1$ be a complex torus and the corresponding integral lattice is (see Fig.~ 3.1)
\begin{align}
\mathcal{L}_\tau  : = \{ a + b\tau \; | \;a,b \in \Z\} ,\quad M: = \mathbb{C} \backslash \mathcal{L}_\tau.
\end{align}
\noindent The relevant eigenvalue problem is
\begin{align}
\Delta _\tau  \phi _\lambda  (x,y) = \lambda ^2 \phi _\lambda  (x,y),
\end{align}
with periodic boundary conditions
\begin{align}
\phi _\lambda  (x,y) = \phi _\lambda  (x + 1,y),\quad \frac{\partial }
{{\partial x}}\phi _\lambda  (x,y) = \frac{\partial }
{{\partial x}}\phi _\lambda  (x + 1,y),
\end{align}
on $x$, as well as
\begin{align}
\phi _\lambda  (x,y) = \phi _\lambda  (x,y + 1),\quad \frac{\partial }
{{\partial y}}\phi _\lambda  (x,y) = \frac{\partial }
{{\partial y}}\phi _\lambda  (x,y + 1),
\end{align}
on $y$. 
\begin{figure}[H]
	\centering
		\includegraphics{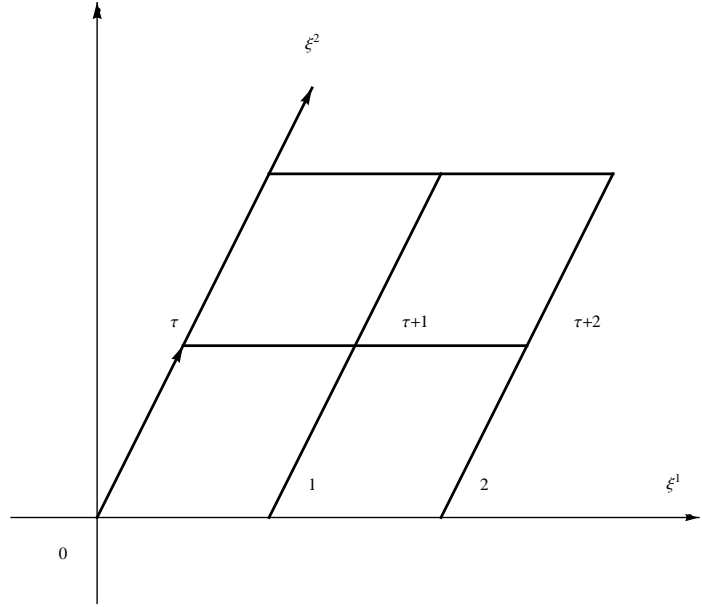}
	\caption{$\tau$ parametrizes the complex structure of this parallelogram \cite{Nakahara}.}
\end{figure}
\noindent By taking the eigenfunctions to be
\begin{align}
\phi _{m,n} (x,y) = e^{ - 2\pi imx} e^{ - 2\pi iny},
\end{align}
with $(m,n) \in \mathbb{Z}_*^2$, we see that
\begin{align}
\Delta _\tau  \phi _{m,n} (x,y) = \frac{{(2\pi )^2 }}
{{\tau _2^2 }}[m^2  + 2\tau _1 mn + (\tau _1^2  + \tau _2^2 )n^2 ]\phi _{m,n} (x,y) = \frac{{(2\pi )^2 }}
{{\tau _2^2 }}\left| {m + n\tau } \right|^2 \phi _{m,n} (x,y).
\end{align}
Therefore, the eigenvalues to consider are
\begin{align} \label{eiganvalues}
\lambda _{m,n}^2 = \frac{{{{(2\pi )}^2}}}{{\tau _2^2}}{\left| {m + n\tau } \right|^2} = \frac{{{{(2\pi )}^2}}}{{\tau _2^2}}(m + n\tau )(m + n\bar \tau ),\quad (m,n) \in \Z_*^2,
\end{align}
and we define the following spectral function.
\begin{definition}
For $\operatorname{Re} (s) > 1$ and $\tau \in \mathbb{H}$, the associated spectral zeta function of the general $\tau$-Laplacian on the complex torus is defined to be
\begin{align}
  \zeta _{\Delta _\tau  } (s) := \sum\limits_\lambda  {(\lambda ^2 )^{ - s} } &= (2\pi )^{ - 2s} \tau _2^{2s} \sum\limits_{(m,n) \in \Z_*^2 } {\left| {m + n\tau } \right|^{ - 2s} }  \nonumber \\
   &= (2\pi )^{ - 2s} \tau _2^s E^* (s,\tau ) ,
\end{align}
where $E^* (s,\tau )$ is the non-holomorphic Eisenstein series.
\end{definition}
\section{Main result}
\noindent A majority of the cases treated in the literature of spectral zeta functions \cite{Elizalde,EORBZ,Kirsten4c} have eigenvalues which give rise to homogeneous Epstein zeta functions, of the type \eqref{homogeneousEpstein},
with no mixed terms $mn$ (or equivalently, where there is no mixed partial derivative
$\partial ^2 /\partial x\partial y$ in the Laplacian). In particular, the inhomogeneous Epstein zeta function
$\zeta_\mathcal{E} (s) = \sum\nolimits_{(m,n) \in \Z_*^2 } {Q(m,n)^{ - s} }$ for a general quadratic form $Q(m,n) = am^2  + bmn + cn^2 $ with $b \ne 0$ has not been computed with the argument principle, only with Poisson summation. This therefore constitutes a new application of the contour integration method which does not exist in the literature and where different
insights are gained (in \cite{Siegel} contour integration is used but the essential step is accomplished through Fourier methods). More general cases with $Q'(m,n) = am^2  + bmn + cn^2  + dm+en+f$, where $f$ is a real positive constant, and in higher dimensions as well,
have been treated in \cite{Elizalde} by means of Poisson summation.\\\\
We split the summation into $n=0$, $m \in \Z \backslash \{ 0 \}$ and $n \ne 0$, $m \in \Z$. Thus we write
\begin{align}
  \zeta _{\Delta _\tau }(s) &= {(2\pi )^{ - 2s}}\tau _2^{2s}\sum\limits_{m =  - \infty }^\infty  {\!\!\!\!^\prime}\,\,\,m^{ - 2s} + {(2\pi )^{ - 2s}}\tau _2^{2s}
  \sum\limits_{n =  - \infty }^\infty  {\!\!\!\!^\prime}\,\,\,\sum\limits_{m =  - \infty }^\infty  \left[ (m + n\tau )(m + n\bar \tau ) \right]^{ - s}  \nonumber \\
  &  :=  2{(2\pi )^{ - 2s}}\tau _2^{2s}{\zeta _R}(2s) + {(2\pi )^{ - 2s}}\tau _2^{2s}{\zeta _I}(s), \nonumber
\end{align}
where $\zeta_R(s)$ denotes the Riemann zeta function and
\[{\zeta _I}(s) = \sum\limits_{n =  - \infty }^\infty  {\!\!\!\!^\prime}\,\,\,\,\sum\limits_{m =  - \infty }^\infty  {{{\left[ {(m + n\tau )(m + n\bar \tau )} \right]}^{ - s}}} . \]
We represent the zeta function $\zeta_I(s)$ in terms of a contour integral; the summation over $m$ is expressed using $\sin(\pi k) = 0$. Thus, we write
\[{\zeta _I}(s) = \sum\limits_{n =  - \infty }^\infty  {\!\!\!\!^\prime}\,\,\,\int_\gamma  {\frac{{dk}}{{2\pi i}}{{[(k + n\tau )(k + n\bar \tau )]}^{ - s}}\frac{d}{{dk}}\log \sin (\pi k)} , \]
where $\gamma$ is the contour in Fig. 4.1.
\begin{figure}[H]
	\centering
		\includegraphics{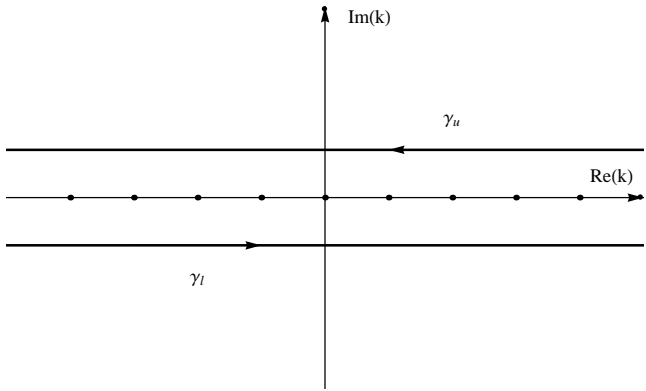}
	\caption{Contour $\gamma$ enclosing the eigenvalues.}
\end{figure}
\noindent When deforming the contour, we need to know the location of the branch cuts. So we solve for
\[(k + n\tau )(k + n\bar \tau ) =  - \alpha , \]
where $\alpha \in \R$ and $\alpha \ge 0$. Writing $\tau = \tau_1 + i \tau_2$, we have
\begin{align} \label{quadratic_equation}
{k^2} + {n^2}{\left| \tau  \right|^2} + kn\bar \tau  + kn\tau  + \alpha  = 0,
\end{align}
so that
\[k =  - n{\tau _1} \pm \sqrt {{n^2}(\tau _1^2 - {{\left| \tau  \right|}^2}) - \alpha }  =  - n{\tau _1} \pm \sqrt { - {n^2}\tau _2^2 - \alpha } .\]
We notice that
\[ - {n^2}\tau _2^2 - \alpha  \leqslant 0,\]
and the branch cuts are something like the ones given in Fig. 4.2.
\begin{figure}[H]
	\centering
		\includegraphics{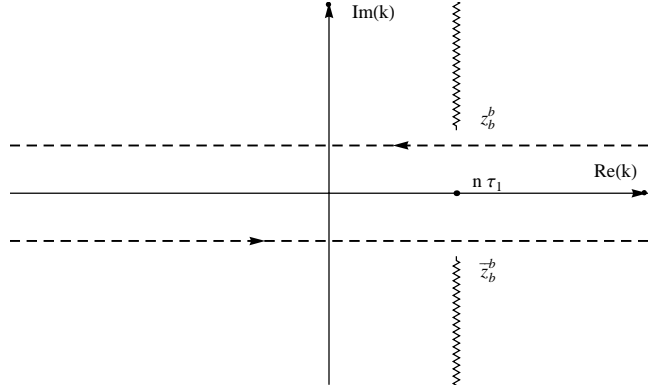}
	\caption{Location of the branch points}
\end{figure}
\noindent Let us denote the branch points by $z_b^n$ and $\bar z_b^n$, thus
\[z_b^n =  - n{\tau _1} + i\sqrt {{n^2}\tau _2^2}  =  - n{\tau _1} + i\left| n \right|{\tau _2}.\]
The natural deformation is, therefore, as indicated in Fig. 4.3.
\begin{figure}[H]
	\centering
		\includegraphics{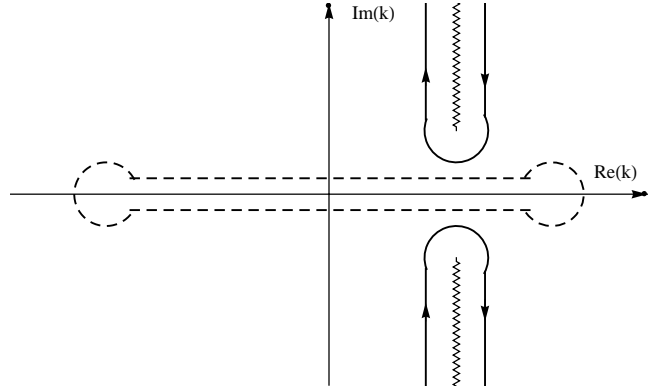}
	\caption{Newly deformed contour enclosing the branch cuts}
\end{figure}
\noindent When shrinking the contours to the branch cuts, the parametrizations will be done according to the following:
\begin{enumerate}
\item for the upper contour
\[\begin{gathered}
  k = z_b^n + {e^{i\pi /2}}u,\quad u \in (\infty ,0], \hfill \\
  k = z_b^n + {e^{ - 3i\pi /2}}u,\quad u \in [0,\infty ); \hfill \\
\end{gathered} \]
\item similarly for the lower contour
\[\begin{gathered}
  k = \bar z_b^n + {e^{ - i\pi /2}}u,\quad u \in [0,\infty ), \hfill \\
  k = \bar z_b^n + {e^{3i\pi /2}}u,\quad u \in (\infty ,0]. \hfill \\
\end{gathered} \]
\end{enumerate}
For $\zeta_I(s)$ this gives
\begin{align}
  \zeta _I(s) &= \sum\limits_{n =  - \infty }^\infty  {\!\!\!\!^\prime}
  \left\{ \int_\infty ^0 \frac{du}{2\pi i} [(z_b^n + e^{i\pi /2}u + n\tau )(z_b^n + e^{i\pi /2}u + n\bar \tau )]^{-s} \frac{d}{du}\log \sin (\pi [z_b^n + e^{i\pi /2}u]) \right.  \nonumber \\
   &+ \int_0^\infty  {\frac{{du}}{{2\pi i}}{{[(z_b^n + {e^{ - 3i\pi /2}}u + n\tau )(z_b^n + {e^{ - 3i\pi /2}}u + n\bar \tau )]}^{ - s}}\frac{d}{{du}}\log \sin (\pi [z_b^n + {e^{ - 3i\pi /2}}u])}  \nonumber \\
   &+ \int_0^\infty  {\frac{{du}}{{2\pi i}}{{[(\bar z_b^n + {e^{ - i\pi /2}}u + n\tau )(\bar z_b^n + {e^{-i\pi /2}}u + n\bar \tau )]}^{ - s}}\frac{d}{{du}}\log \sin (\pi [\bar z_b^n + {e^{ - i\pi /2}}u])}  \nonumber \\
  &\left. { + \int_\infty ^0 {\frac{{du}}{{2\pi i}}{{[(\bar z_b^n + {e^{3i\pi /2}}u + n\tau )(\bar z_b^n + {e^{3i\pi /2}}u + n\bar \tau )]}^{ - s}}\frac{d}{{du}}\log \sin (\pi [\bar z_b^n + {e^{3i\pi /2}}u])} } \right\}. \nonumber
\end{align}
We next rewrite the integrands using the fact that $z_b^n$ and $\bar z_b^n$ solve the quadratic equation \eqref{quadratic_equation} with $\alpha =0$. We  then use the notation
\[{\zeta _I}(s) = \zeta _I^{(1)}(s) + \zeta _I^{(2)}(s) + \zeta _I^{(3)}(s) + \zeta _I^{(4)}(s)\]
to denote each individual (infinite) sum above. Let us start with $\zeta _I^{(1)}(s)$: we begin by computing
\[(z_b^n + {e^{i\pi /2}}u + n\tau )(z_b^n + {e^{i\pi /2}}u + n\bar \tau ) = {e^{i\pi }}({u^2} + 2u\left| n \right|{\tau _2}),\]
so that we have
\[\zeta _I^{(1)}(s) = \sum\limits_{n =  - \infty }^\infty  {\!\!\!\!^\prime}\,\,\,
{( - {e^{ - i\pi s}})\int_0^\infty  {\frac{{du}}{{2\pi i}}{{({u^2} + 2u\left| n \right|{\tau _2})}^{ - s}}\frac{d}{{du}}\log \sin (\pi [z_b^n + iu])} } .\]
Redoing the same computation for $\zeta_I^{(2)}(s)$, but replacing $e^{i \pi /2}$ with $e^{-3 i \pi /2}$ yields
\[(z_b^n + {e^{ - 3i\pi /2}}u + n\tau )(z_b^n + {e^{ - 3i\pi /2}}u + n\bar \tau ) = {e^{ - i\pi }}({u^2} + 2u\left| n \right|{\tau _2}),\]
so that
\[\zeta _I^{(2)}(s) = \sum\limits_{n =  - \infty }^\infty  {\!\!\!\!^\prime}
\,\,\,{{e^{i\pi s}}\int_0^\infty  {\frac{{du}}{{2\pi i}}{{({u^2} + 2u\left| n \right|{\tau _2})}^{ - s}}\frac{d}{{du}}\log \sin (\pi [z_b^n + iu])} } \]
and, combining the two terms,
\[\zeta _I^{(1)}(s) + \zeta _I^{(2)}(s) = \frac{{\sin (\pi s)}}{\pi }\sum\limits_{n =  - \infty }^\infty  {\!\!\!\!^\prime} \,\,\,
\int_0^\infty  du ({u^2} + 2u\left| n \right|{\tau _2})^{ - s}\frac{d}{{du}}\log \sin (\pi [z_b^n + iu]). \]
We next consider the $\log$ terms in order to perform the analytic continuation. A suitable rewriting is
\[\sin (\pi [z_b^n + iu]) = \frac{{{e^{i\pi (iu + z_b^n)}} - {e^{ - i\pi (iu + z_b^n)}}}}{{2i}} =  - \frac{1}{{2i}}{e^{\pi u - i\pi z_b^n}}(1 - {e^{ - 2\pi u + 2i\pi z_b^n}}).\]
Note that ${e^{2i\pi z_b^n}} = {e^{2i\pi ( - n{\tau _1} + i\left| n \right|{\tau _2})}}$ is exponentially damped for large $\left| n \right|$. We therefore write
\begin{align}
  \zeta _I^{(1)}(s) + \zeta _I^{(2)}(s) &= \frac{{\sin (\pi s)}}{\pi }\sum\limits_{n =  - \infty }^\infty  {\!\!\!\!^\prime}
  \,\,\,{\int_0^\infty  {du{{({u^2} + 2u\left| n \right|{\tau _2})}^{ - s}}\frac{d}{{du}}\log [{e^{\pi u - i\pi z_b^n}}(1 - {e^{ - 2\pi u + 2i\pi z_b^n}})]} }  \nonumber \\
   &= \frac{{\sin (\pi s)}}{\pi }\sum\limits_{n =  - \infty }^\infty  {\!\!\!\!^\prime}\,\,\,{\int_0^\infty  {du{{({u^2} + 2u\left| n \right|{\tau _2})}^{ - s}}\left\{ {\pi  + \frac{d}{{du}}\log (1 - {e^{ - 2\pi u + 2i\pi z_b^n}})} \right\}} }  \nonumber \\
   &=: \zeta _I^{(12,1)}(s) + \zeta _I^{(12,2)}(s), \nonumber
\end{align}
where we define
\[\zeta _I^{(12,1)}(s)  :=  \sin (\pi s)\sum\limits_{n =  - \infty }^\infty  {\!\!\!\!^\prime}\,\,\,{\int_0^\infty  {du{{({u^2} + 2u\left| n \right|{\tau _2})}^{ - s}}} }, \]
and
\[\zeta _I^{(12,2)}(s)  :=  \frac{{\sin (\pi s)}}{\pi }\sum\limits_{n =  - \infty }^\infty  {\!\!\!\!^\prime}
\,\,\,{\int_0^\infty  {du{{({u^2} + 2u\left| n \right|{\tau _2})}^{ - s}}\frac{d}{{du}}\log (1 - {e^{ - 2\pi u + 2i\pi z_b^n}})} }. \]
To compute the derivative at $s=0$ we only need to find the analytical continuation of $\zeta _I^{(12,1)}(s)$ since $\zeta _I^{(12,2)}(s)$
is already valid for all $s \in \C$. In order to accomplish this continuation we note that, for $\tfrac{1}{2} < \operatorname{Re}(s) < 1$,
\begin{align} \label{gammafunctionintegral}
\int_0^\infty du\,\, {{u^{ - s}}{{(u + 2x)}^{ - s}}}  = {x^{1 - 2s}}\frac{{\Gamma (1 - s)\Gamma (s - \tfrac{1}{2})}}{{2\sqrt \pi  }}.
\end{align}
With this in mind, we write
\begin{align}
  \zeta _I^{(12,1)}(s) &= \sin (\pi s)\sum\limits_{n =  - \infty }^\infty  {\!\!\!\!^\prime}\,\,\,
  {\int_0^\infty  {du{{({u^2} + 2u\left| n \right|{\tau _2})}^{ - s}}} }  = 2\sin (\pi s)\sum\limits_{n = 1}^\infty  {\int_0^\infty  {du{u^{ - s}}{{(u + 2n{\tau _2})}^{ - s}}} }  \nonumber \\
   &= 2\sin (\pi s)\frac{{\Gamma (1 - s)\Gamma (s - \tfrac{1}{2})}}{{2\sqrt \pi  }}\tau _2^{1 - 2s}\sum\limits_{n = 1}^\infty  {{n^{1 - 2s}}}  = \frac{{\sqrt \pi  \Gamma (s - \tfrac{1}{2})}}{{\Gamma (s)}}\tau _2^{1 - 2s}{\zeta _R}(2s - 1). \nonumber
\end{align}
The branch in the lower half-plane is handled accordingly. In order to simplify the integrand, we note the analogy between the first and third case: $e^{i \pi /2} \to e^{-i \pi/2}$, $z_b^n \to \bar z_b^n$. Thus
\[(\bar z_b^n + {e^{ - i\pi /2}}u + n\tau )(\bar z_b^n + {e^{ - i\pi /2}}u + n\bar \tau ) = {e^{ - i\pi }}({u^2} + 2u\left| n \right|{\tau _2}),\]
so that the third function is
\[\zeta _I^{(3)}(s) = \sum\limits_{n =  - \infty }^\infty  {\!\!\!\!^\prime}\,\,\,
{({e^{i\pi s}})\int_0^\infty  {\frac{{du}}{{2\pi i}}{{({u^2} + 2u\left| n \right|{\tau _2})}^{ - s}}\frac{d}{{du}}\log \sin (\pi [\bar z_b^n - iu])} }. \]
Similarly, for the last function it follows that
\[\zeta _I^{(4)}(s) = \sum\limits_{n =  - \infty }^\infty  {\!\!\!\!^\prime}\,\,\,
{( - {e^{ - i\pi s}})\int_0^\infty  {\frac{{du}}{{2\pi i}}{{({u^2} + 2u\left| n \right|{\tau _2})}^{ - s}}\frac{d}{{du}}\log \sin (\pi [\bar z_b^n - iu])} }. \]
Adding up these two terms yields
\[\zeta _I^{(3)}(s) + \zeta _I^{(4)}(s) = \frac{{\sin (\pi s)}}{\pi }\sum\limits_{n =  - \infty }^\infty  {\!\!\!\!^\prime}\,\,\,
{\int_0^\infty  {du{{({u^2} + 2u\left| n \right|{\tau _2})}^{ - s}}\frac{d}{{du}}\log \sin (\pi [\bar z_b^n - iu])} }. \]
Going through a similar manipulation of the $\log \sin $ term as above allows us to write
\[\sin (\pi [\bar z_b^n - iu]) = \frac{{{e^{i\pi (\bar z_b^n - iu)}} - {e^{ - i\pi (\bar z_b^n - iu)}}}}{{2i}} = \frac{1}{{2i}}{e^{\pi u + i\pi \bar z_b^n}}(1 - {e^{ - 2\pi u - 2i\pi \bar z_b^n}}).\]
We note that $\pi u + i\pi \bar z_b^n = \pi u + i\pi ( - n{\tau _1} - i\left| n \right|{\tau _2})$ and so ${e^{ - 2\pi u - 2i\pi \bar z_b^n}}$ is, like in the previous case, exponentially damped as $\left| n \right| \to \infty$. Once more, using \eqref{gammafunctionintegral},
\begin{align}
  \zeta _I^{(3)}(s) + \zeta _I^{(4)}(s) &= \frac{{\sin (\pi s)}}{\pi }\sum\limits_{n =  - \infty }^\infty  {\!\!\!\!^\prime}\,\,\,
  {\int_0^\infty  {du{{({u^2} + 2u\left| n \right|{\tau _2})}^{ - s}}\left[ {\pi  + \frac{d}{{du}}\log (1 - {e^{ - 2\pi u - 2i\pi \bar z_b^n}})} \right]} }  \nonumber \\
   &= \frac{{\sqrt \pi  \Gamma (s - \tfrac{1}{2})}}{{\Gamma (s)}}\tau _2^{1 - 2s}{\zeta _R}(2s - 1) \nonumber \\
   &+ \frac{{\sin (\pi s)}}{\pi }\sum\limits_{n =  - \infty }^\infty  {\!\!\!\!^\prime}\,\,\,
   {\int_0^\infty  {du{{({u^2} + 2u\left| n \right|{\tau _2})}^{ - s}}\frac{d}{{du}}\log (1 - {e^{ - 2\pi u - 2i\pi \bar z_b^n}})} }.  \nonumber
\end{align}
Therefore, the final result is
\begin{align}
  {\zeta _{{\Delta _\tau }}}(s) &= 2{(2\pi )^{ - 2s}}\tau _2^{2s}{\zeta _R}(2s) + {(2\pi )^{1 - 2s}}\tau _2^{2s}\frac{{\Gamma (s - \tfrac{1}{2})}}{{\sqrt \pi  \Gamma (s)}}\tau _2^{1 - 2s}{\zeta _R}(2s - 1) \nonumber \\
   &+ {(2\pi )^{ - 2s}}\tau _2^{2s}\frac{{\sin (\pi s)}}{\pi }\sum\limits_{n =  - \infty }^\infty  {\!\!\!\!^\prime}\,\,\,
   {\int_0^\infty  {du{u^{ - s}}{{(u + 2\left| n \right|{\tau _2})}^{ - s}}}} \nonumber\\
   &\quad \quad\quad \quad {\times\frac{d}{{du}}\log [(1 - {e^{ - 2\pi u + 2i\pi z_b^n}})(1 - {e^{ - 2\pi u + 2i\pi \bar z_b^n}})]} , \nonumber
\end{align}
which is now valid for all $s \in \C \backslash\{1\}$. Re-writing the sum so that it goes from $n=1$ to $n = \infty$ we have thus proved the following result.
\begin{proposition}
The spectral zeta function of $\Delta_{\tau}$ on $S^1 \times S^1$ can be written as
\begin{align} \label{main_theorem}
  {\zeta _{{\Delta _\tau }}}(s) &= 2{(2\pi )^{ - 2s}}\tau _2^{2s}{\zeta _R}(2s) + {(2\pi )^{1 - 2s}}{\tau _2}\frac{{\Gamma (s - \tfrac{1}{2})}}{{\sqrt \pi  \Gamma (s)}}{\zeta _R}(2s - 1) \nonumber \\
   &+ \frac{{2\sin (\pi s)}}{\pi }{\left( {\frac{{2\pi }}{{{\tau _2}}}} \right)^{ - 2s}}\sum\limits_{n = 1}^\infty  {\int_0^\infty  {du{{({u^2} + 2un{\tau _2})}^{ - s}}} \frac{d}{{du}}\log [(1 - {e^{ - 2\pi u - 2i\pi n\bar \tau }})(1 - {e^{ - 2\pi u + 2i\pi n\tau }})]},
\end{align}
for $\tau \in \mathbb{H}$ and $s \in \C \backslash\{1\}$.
\end{proposition}
\noindent Before proceeding to explain the term on the second line of \eqref{main_theorem}, we first compute the functional determinant we were interested in. The derivative of the last expression at $s=0$ is obtained in terms of the Dedekind eta function:
 \begin{align}
  {\zeta' _{{\Delta _\tau }}}(0) &=  - \log \tau _2^2 + \frac{{\pi {\tau _2}}}{3} + 2\sum\limits_{n = 1}^\infty  {\int_0^\infty  {du} \frac{d}{{du}}\log [(1 - {e^{ - 2\pi u - 2i\pi n\bar \tau }})(1 - {e^{ - 2\pi u + 2i\pi n\tau }})]}  \nonumber \\
   &=  - \log \tau _2^2 + \frac{{\pi {\tau _2}}}{3} - 2\sum\limits_{n = 1}^\infty \left[ {\log (1 - {e^{ - 2i\pi n\bar \tau }}) + \log (1 - {e^{ - 2\pi u + 2i\pi n\tau }})} \right] \nonumber \\
   &=  - \log \tau _2^2 + \frac{{\pi {\tau _2}}}{3} - 2\left( {\log \eta (\tau ) - \frac{{\pi i\tau }}{{12}} + \log \eta ( - \bar \tau ) + \frac{{\pi i\bar \tau }}{{12}}} \right) \nonumber \\
   &=  - \log \tau _2^2 + \frac{{\pi {\tau _2}}}{3} - 2\left( {\log [\eta (\tau )\eta ( - \bar \tau )] + \frac{{\pi i}}{{12}}(\bar \tau  - \tau )} \right) \nonumber \\
   &=  - \log (\tau _2^2{\left| {\eta (\tau )} \right|^4}), \nonumber
\end{align}
since $\eta ( - \bar \tau ) = \overline {\eta (\tau )}$. Therefore, we have the following result \cite{Deligne, DMS, Polchinski, Williams1}.
\begin{theorem}
The functional determinant of the $\tau$-Laplacian on the complex torus is
\begin{align} \label{functional_det}
\det ({\Delta _\tau }) = \tau _2^2{\left| {\eta (\tau )} \right|^4}.
\end{align}
\end{theorem}
\begin{proof}
By definition
\[\det ({\Delta _\tau }) = \exp ( - {\zeta' _{{\Delta _\tau }}}(0)) = \tau _2^2{\left| {\eta (\tau )} \right|^4},\]
as claimed.
\end{proof}
\noindent We note that in order to obtain the value of $\zeta'_{\Delta_\tau}(0)$ we have not used Kronecker's first limit formula, which in turn, depends on the functional equation of $E^*(s,\tau)$. Thus, viewed under this optic, the method of contour integration is cheaper in the sense that it requires less resources to provide the functional determinant. For a derivation of this theorem using the Kronecker formula, see, for instance, \cite{Williams1}.
\section{Additional results and consequences}
\subsection{The Chowla-Selberg series formula}
In fact, the term on the second line of \eqref{main_theorem} can be shown to be the term from the Chowla-Selberg series formula obtained through Poisson summation methods. To see this, let us compute the $\log$-terms further. Expanding the logarithm,
\[\log (1 - {e^{ - 2\pi u - 2i\pi n\bar \tau }}) =  - \sum\limits_{k = 1}^\infty  {\frac{1}{k}{e^{ - 2\pi ku - 2i\pi nk\bar \tau }}}, \]
and for $\operatorname{Re}(s)<1$ the integral becomes
\begin{align}
  {\Upsilon _{\bar \tau }}(s,n) & : =  \int_0^\infty  {du{u^{ - s}}{{(u + 2n{\tau _2})}^{ - s}}\frac{d}{{du}}\log [(1 - {e^{ - 2\pi u - 2i\pi n\bar \tau }})]}  \nonumber \\
   &=  - \int_0^\infty  {du{u^{ - s}}{{(u + 2n{\tau _2})}^{ - s}}\frac{d}{{du}}\sum\limits_{k = 1}^\infty  {\frac{1}{k}{e^{ - 2\pi ku - 2i\pi nk\bar \tau }}} }  \nonumber \\
   &= 2\pi \sum\limits_{k = 1}^\infty  {\int_0^\infty  {du{u^{ - s}}{{(u + 2n{\tau _2})}^{ - s}}{e^{ - 2\pi ku - 2i\pi nk\bar \tau }}} }  \nonumber \\
   &= 2\pi \sum\limits_{k = 1}^\infty  {{e^{ - 2i\pi nk\bar \tau }}\int_0^\infty  {du{u^{ - s}}{{(u + 2n{\tau _2})}^{ - s}}{e^{ - 2\pi ku}}} }  \nonumber \\
   &= 2\pi \sum\limits_{k = 1}^\infty  {{e^{ - 2i\pi nk\bar \tau }}\frac{{{K_{1/2 - s}}(2\pi kn{\tau _2})}}{{\sin (\pi s)\Gamma (s)}}{e^{2\pi nk{\tau _2}}}{k^{s - 1/2}}{\pi ^s}\frac{1}{{\sqrt {n{\tau _2}} }}{{(n{\tau _2})}^{1 - s}}}  \nonumber \\
   &= \frac{{2{\pi ^{1 + s}}\tau _2^{1/2 - s}}}{{\sin (\pi s)\Gamma (s)}}\sum\limits_{k = 1}^\infty  {{e^{ - 2\pi nki{\tau _1}}}{k^{s - 1/2}}{n^{1/2 - s}}{K_{1/2 - s}}(2\pi kn{\tau _2})},  \nonumber
\end{align}
where $K_n(z)$ is the Bessel function of the second kind, defined by
\[{K_\nu }(z) : = \frac{1}{2}\int_0^\infty dt\,\, {\exp \left[ { - \frac{x}{2}\left( {t + \frac{1}{t}} \right)} \right]{t^{\nu  - 1}}}. \]
Moreover, the integral of the other $\log$-term only differs in that $e^{-2 \pi i n k \bar \tau}$ is replaced by $e^{2 \pi i n k \tau}$. Thus, noting $\tau_2 + i \tau = \tau_2 + i(\tau_1 + i \tau_2) = i\tau_1$, we have
\begin{align}
  {\Upsilon _\tau }(s,n) &  := \int_0^\infty  {du{u^{ - s}}{{(u + 2n{\tau _2})}^{ - s}}\frac{d}{{du}}\log [(1 - {e^{ - 2\pi u + 2i\pi n\tau }})]}  \nonumber \\
   &= \frac{{2{\pi ^{1 + s}}\tau _2^{1/2 - s}}}{{\sin (\pi s)\Gamma (s)}}\sum\limits_{k = 1}^\infty  {{e^{2\pi nki{\tau _1}}}{k^{s - 1/2}}{n^{1/2 - s}}{K_{1/2 - s}}(2\pi kn{\tau _2})}.  \nonumber
\end{align}
If we set the third term in \eqref{main_theorem} to be $Q(s,\tau)$, then we see that
\begin{align}
  Q(s,\tau ) &= \frac{{2\sin (\pi s)}}{\pi }{\left( {\frac{{2\pi }}{{{\tau _2}}}} \right)^{ - 2s}}\sum\limits_{n = 1}^\infty  {\left[{\Upsilon _\tau }(s,n) + {\Upsilon _{\bar \tau }}(s,n)\right]}  \nonumber \\
   &= \frac{{2\sin (\pi s)}}{\pi }{\left( {\frac{{2\pi }}{{{\tau _2}}}} \right)^{ - 2s}}\frac{{2{\pi ^{1 + s} \tau_2^{1/2-s} }}}{{\sin (\pi s)\Gamma (s)}}\sum\limits_{n = 1}^\infty  {\sum\limits_{k = 1}^\infty  {{e^{2\pi nki{\tau _1}}}{k^{s - 1/2}}{n^{1/2 - s}}{K_{1/2 - s}}(2\pi kn{\tau _2})} }  \nonumber \\
   &+ \frac{{2\sin (\pi s)}}{\pi }{\left( {\frac{{2\pi }}{{{\tau _2}}}} \right)^{ - 2s}}\frac{{2{\pi ^{1 + s} \tau_2^{1/2-s} }}}{{\sin (\pi s)\Gamma (s)}}\sum\limits_{n = 1}^\infty  {\sum\limits_{k = 1}^\infty  {{e^{ - 2\pi nki{\tau _1}}}{k^{s - 1/2}}{n^{1/2 - s}}{K_{1/2 - s}}(2\pi kn{\tau _2})} }  \nonumber \\
   &= \frac{{{2^{3 - 2s}}{\pi ^{ - s}}\tau _2^{1/2 + s}}}{{\Gamma (s)}}\sum\limits_{n = 1}^\infty  {\sum\limits_{k = 1}^\infty
    {\cos (2\pi nk{\tau _1}){{\left( {\frac{k}{n}} \right)}^{s - 1/2}}{K_{1/2 - s}}(2\pi kn{\tau _2})} }.  \nonumber
\end{align}
The key is now to relate the expression inside the double sum to the divisor function. The main property we are interested in is the following way of changing the double sum for a single sum while bringing in the divisor function
\[\sum\limits_{n = 1}^\infty  {\left[ {\sum\limits_{k = 1}^\infty  {{e^{ \pm 2\pi ikn{\tau _1}}}{K_{1/2 - s}}(2\pi kn{\tau _2}){{\left( {\frac{k}{n}} \right)}^{s - 1/2}}} } \right] =
\sum\limits_{n = 1}^\infty  {{\sigma _{1 - 2s}}(n){e^{ \pm 2\pi in{\tau _1}}}{K_{1/2 - s}}(2\pi n{\tau _2}){n^{s - 1/2}}} } . \]
For a proof of a more general result of this type, see for instance \cite{Williams1}. Thus, switching back to $E^*(s,\tau)$ instead of $\zeta_{\Delta_\tau}(s)$, we have arrived at the following theorem,
which now holds for all $s \ne 1$ by analytic continuation \cite{ChowlaSelberg}.
\begin{theorem}[Chowla-Selberg series formula]
One has
\begin{align}
  {E^*}(s,\tau ) &= 2\tau _2^s{\zeta _R}(2s) + 2\pi \tau _2^{1 - s}\frac{{\Gamma (s - \tfrac{1}{2})}}{{\sqrt \pi  \Gamma (s)}}{\zeta _R}(2s - 1) \nonumber \\
   &+ \frac{{8{\pi ^s}\tau _2^{1/2}}}{{\Gamma (s)}}\sum\limits_{n = 1}^\infty  {{\sigma _{1 - 2s}}(n)\cos (2\pi n{\tau _1}){K_{1/2 - s}}(2\pi n{\tau _2}){n^{s - 1/2}}},  \nonumber
\end{align}
for $\tau \in \mathbb{H}$ and $s \in \C \backslash \{ 1 \}$.
\end{theorem}
\subsection{The Nan-Yue Williams formula}
As a byproduct, we may now equate the two `remainders' in the above expressions of $E^*(s,\tau)$, recalling that the one in \eqref{main_theorem} has to be multiplied by $(2\pi)^{2s}\tau_2^{-s}$, to get
\begin{align}
  Q(s,\tau ) &= \frac{{8{\pi ^s}\tau _2^{1/2}}}{{\Gamma (s)}}\sum\limits_{n = 1}^\infty  {{\sigma _{1 - 2s}}(n)\cos (2\pi n{\tau _1}){K_{1/2 - s}}(2\pi n{\tau _2}){n^{s - 1/2}}}  \nonumber \\
   &= 2\tau _2^s\frac{{\sin (\pi s)}}{\pi }\sum\limits_{n = 1}^\infty  {\int_0^\infty  {du{{({u^2} + 2un{\tau _2})}^{ - s}}} \frac{d}{{du}}\log [(1 - {e^{ - 2\pi u - 2i\pi n\bar \tau }})(1 - {e^{ - 2\pi u + 2i\pi n\tau }})]}.  \nonumber
\end{align}
Using the Euler reflection formula for the $\Gamma$-function we obtain
\begin{align} \label{equation_corollaries}
  &\sum\limits_{n = 1}^\infty  {{\sigma _{1 - 2s}}(n)\cos (2\pi n{\tau _1}){K_{1/2 - s}}(2\pi n{\tau _2}){n^{s - 1/2}}}  \nonumber \\
  & = \frac{{\tau _2^{ - 1/2}}}{4}{\left( {\frac{{{\tau _2}}}{\pi }} \right)^s}\frac{1}{{\Gamma (1 - s)}}\sum\limits_{n = 1}^\infty  {\int_0^\infty  {du{{({u^2} + 2un{\tau _2})}^{ - s}}} \frac{d}{{du}}\log [(1 - {e^{ - 2\pi u - 2i\pi n\bar \tau }})(1 - {e^{ - 2\pi u + 2i\pi n\tau }})]}.
\end{align}
At $s=0$ the integral is easily evaluated by a similar computation to the one done previously for $\zeta'_{\Delta_{\tau}}(0)$, so that we are left with the following consequence.
\begin{corollary}
For $\tau \in \mathbb{H}$, one has
\begin{align}
  \sum\limits_{n = 1}^\infty  {{\sigma _1}(n)\cos (2\pi n{\tau _1}){K_{1/2}}(2\pi n{\tau _2}){n^{ - 1/2}}}  =  - \frac{{\tau _2^{ - 1/2}}}{2}\log \left| {\eta (\tau )} \right| - \frac{{\tau _2^{1/2}\pi }}{{24}}. \nonumber
\end{align}
\end{corollary}
\noindent In \cite{NanWilliams} this is given in a slightly different context and
\[\sum\limits_{n = 1}^\infty  {{\sigma _1}(n){K_{1/2}}(2\pi n){n^{ - 1/2}}}  =  - \frac{1}{2}\log \eta (i) - \frac{\pi }{{24}} \approx 0.000936341.\]
for the particular case $\tau=i$.
\subsection{The Kronecker first limit formula}
Let us now go back to the Chowla-Selberg formula. Using the functional equation of the divisor function and the Bessel function, respectively,
\[{\sigma _\nu }(n) = {n^\nu }{\sigma _{ - \nu }}(n),\quad {K_\nu }(x) = {K_{ - \nu }}(x),\]
it is not difficult to show that the following functional equation for $E^*(s,\tau)$ holds
\begin{align} \label{functional_nonhom}
{\pi ^{1 - 2s}}\Gamma (s){E^*}(s,\tau ) = \Gamma (1 - s){E^*}(1 - s,\tau ).
\end{align}
Using this functional equation and the result of the determinant of the Laplacian \eqref{functional_det} we can now reverse the steps of the proof of Kronecker's first limit formula \cite{Williams1}. All known proofs (e.g. \cite{Motohashi,NanWilliams,Siegel,Shintani,Williams1}) use a variation of some kind of Poisson summation. The present one is a new proof in the sense that no techniques from Fourier analysis are ever used.\\
Interesting steps in the direction of special functions and complex integration were worked out in \cite{Shintani,Vardi} using the theory of multiple Gamma functions and Barnes zeta functions. Furthermore, using the Barnes double gamma function and the Selberg zeta function,
the determinants of the $n$-sphere and spinor fields on a Riemann surface are found in \cite{Sarnak}; see also \cite{136,138}.
\begin{proof}[Proof of Kronecker's first limit formula]
First we define
\begin{align}
f(s): = \Gamma (s)\pi ^{2 - 2s} /\Gamma (2 - s),
\end{align}
so that
\begin{align}
f(1) = 1,\quad \mathop {\lim }\limits_{s \to 1} \frac{{f(s) - f(1)}}{{s - 1}} = f'(1) =  - 2\log \pi  - 2\gamma,
\end{align}
where $\gamma$ is Euler's constant. The key idea is to use the result we already got
\begin{align}
\zeta _{\Delta _\tau  } '(0) =  - \log \tau _2^2 \left| {\eta (\tau )} \right|^4,
\end{align}
from where
\begin{align} \label{compare1}
  \left. {\frac{\partial }
{{\partial s}}} \right|_{s = 0} E^* (s,\tau ) &= \left. {\frac{{(2\pi )^{2s} }}
{{\tau _2^s }}} \right|_{s = 0} \zeta _{\Delta _\tau  } '(0) + \left. {\frac{\partial }
{{\partial s}}} \right|_{s = 0} \left( {\frac{{(2\pi )^{2s} }}
{{\tau _2^s }}} \right)\zeta _{\Delta _\tau  } (0) \nonumber \\
   &=  - \log (\tau _2^2 \left| {\eta (\tau )} \right|^4)  - (\log 4\pi ^2  - \log \tau _2 ) \nonumber \\
   &=  - \log (4\pi ^2 \tau _2 \left| {\eta (\tau )} \right|^4),
\end{align}
and confront it with the definition of the derivative of $E^* (s,\tau )$,
\begin{align} \label{compare2}
  \left. {\frac{\partial }
{{\partial s}}} \right|_{s = 0} E^* (s,\tau ) &:= \mathop {\lim }\limits_{s \to 0} \frac{{E^* (s,\tau ) - E^* (0,\tau )}}
{{s - 0}} = \mathop {\lim }\limits_{s \to 1} \frac{{E^* (1 - s,\tau ) + 1}}
{{1 - s}} \nonumber \\
   &= \mathop {\lim }\limits_{s \to 1} \left( {\frac{{E^* (1 - s,\tau )\Gamma (s)}}
{{\Gamma (s)(1 - s)}} - \frac{1}
{{s - 1}}} \right) = \mathop {\lim }\limits_{s \to 1} \left( {\frac{{\Gamma (s)}}
{{1 - s}}\pi ^{1 - 2s} \frac{{E^* (s,\tau )}}
{{\Gamma (1 - s)}} - \frac{1}
{{s - 1}}} \right) \nonumber \\
  &= \mathop {\lim }\limits_{s \to 1} \left[ {\frac{{\Gamma (s)\pi ^{1 - 2s} }}
{{\Gamma (2 - s)}}\left( {E^* (s,\tau ) - \frac{\pi }
{{s - 1}}} \right) + \frac{{\Gamma (s)\pi ^{1 - 2s} }}
{{\Gamma (2 - s)}}\frac{\pi }
{{s - 1}} - \frac{1}
{{s - 1}}} \right] \nonumber \\
  & = \mathop {\lim }\limits_{s \to 1} \frac{{\Gamma (s)\pi ^{1 - 2s} }}
{{\Gamma (2 - s)}}\left( {E^* (s,\tau ) - \frac{\pi }
{{s - 1}}} \right) + \mathop {\lim }\limits_{s \to 1} \frac{{f(s) - f(1)}}
{{s - 1}} \nonumber \\
   &= \frac{1}
{\pi }\mathop {\lim }\limits_{s \to 1} \left( {E^* (s,\tau ) - \frac{\pi }
{{s - 1}}} \right) - 2\log \pi  - 2\gamma .
\end{align}
Comparing \eqref{compare1} and \eqref{compare2} yields
\begin{align}
\frac{1}
{\pi }\mathop {\lim }\limits_{s \to 1} \left( {E^* (s,\tau ) - \frac{\pi }
{{s - 1}}} \right) - 2\log \pi  - 2\gamma  =  - \log (4\pi ^2 \tau _2 \left| {\eta (\tau )} \right|^4)
\end{align}
and re-arranging
\begin{align}
\mathop {\lim }\limits_{s \to 1} \left( {E^* (s,\tau ) - \frac{\pi }{{s - 1}}} \right) = 2\pi (\gamma  - \log 2 - \log \tau _2^{1/2} \left| {\eta (\tau )} \right|^2 ),
\end{align}
the desired result follows.
\end{proof}
\subsection{The Lambert series}
\noindent We may re-write \eqref{equation_corollaries} as
\begin{align}
  \mathcal{Q}(s,\tau ) & :=  \sum\limits_{n = 1}^\infty  {{\sigma _{1 - 2s}}(n)\cos (2\pi n{\tau _1}){K_{1/2 - s}}(2\pi n{\tau _2}){n^{s - 1/2}}}  \nonumber \\
  & = \frac{{\tau _2^{ - 1/2}}}{4}{\left( {\frac{{{\tau _2}}}{\pi }} \right)^s}\frac{1}{{\Gamma (2 - s)}} \nonumber \\
   &\times \sum\limits_{n = 1}^\infty  {\int_0^\infty  {\frac{{dx}}{{2x + 2n{\tau _2}}}\left[ {\frac{d}{{dx}}{{({x^2} + 2xn{\tau _2})}^{ - s + 1}}} \right]\frac{d}{{dx}}\log [(1 - {e^{ - 2\pi x - 2\pi in\bar \tau }})(1 - {e^{ - 2\pi x + 2\pi in\tau }})]} }  ,\nonumber
\end{align}
so that, after integrating by parts, we have
\begin{align}
  \mathcal{Q}(s,\tau ) &= \frac{{\tau _2^{ - 1/2}}}{4}{\left( {\frac{{{\tau _2}}}{\pi }} \right)^s}\frac{1}{{\Gamma (2 - s)}}\sum\limits_{n = 1}^\infty  {\left[ {\left( {\frac{{2\pi (1 - 2{e^{2\pi in\tau }} + {e^{2\pi in(\tau  + \bar \tau )}})}}{{({e^{2\pi in\tau }} - 1)({e^{2\pi in\bar \tau }} - 1)}}} \right)\frac{{\Gamma (2 - s)\Gamma (s - \tfrac{1}{2})}}{{2\sqrt \pi  }}{{(n{\tau _2})}^{1 - 2s}}} \right.}  \nonumber \\
   &+ \left. {\int_0^\infty dx\,\,  {g(s,x,n,{\tau _2})\frac{{{d^2}}}{{d{x^2}}}\log [(1 - {e^{ - 2\pi x - 2\pi in\bar \tau }})(1 - {e^{ - 2\pi x + 2\pi in\tau }})]} } \right], \nonumber
\end{align}
where
\[g(s,x,n,{\tau _2}) : = {x\,\,_2F_1}\left( {1 - s,s,2 - s, - \frac{x}{{2n{\tau _2}}}} \right){\left( {1 + \frac{x}{{2n{\tau _2}}}} \right)^s}{(x(x + 2n{\tau _2}))^{ - s}},\]
with ${_2F_1}$ the usual hypergeometric function. Setting $s=1$, we have
\[g(1,x,n,{\tau _2}) = \frac{1}{{2n{\tau _2}}}\]
and the integral can, once again, be easily evaluated
\[\mathcal{Q}(1,\tau ) =  - \frac{{\tau _2^{ - 1/2}}}{4}\sum\limits_{n = 1}^\infty  {\frac{1}{n}\frac{{(1 - 2{e^{2\pi in\tau }} + {e^{2\pi in(\tau  + \bar \tau )}})}}{{({e^{2\pi in\tau }} - 1)({e^{2\pi in\bar \tau }} - 1)}}}. \]
The case $\tau = i$ yields
\[\mathcal{Q}(1,i) = \sum\limits_{n = 1}^\infty  {{\sigma _{ - 1}}(n){K_{ - 1/2}}(2\pi n){n^{1/2}}}  = \frac{1}{2}\sum\limits_{n = 1}^\infty  {\frac{1}{n}\frac{1}{{{e^{2\pi n}} - 1}}}  \approx {\text{0}}{\text{.000936341}}.\]
Using the fact that
\[{K_{ \pm 1/2}}(2\pi n) = \frac{1}{{2\sqrt n }}{e^{ - 2\pi n}},\]
we obtain
\[\sum\limits_{n = 1}^\infty  {{\sigma _{ - 1}}(n){e^{ - 2\pi n}}}  = \sum\limits_{n = 1}^\infty  {\frac{1}{n}\frac{1}{{{e^{2\pi n}} - 1}}}, \]
which is a special case of the Lambert series \cite{Siegel}
\[\sum\limits_{n = 1}^\infty  {{\sigma _\alpha }(n){q^n}}  = \sum\limits_{n = 1}^\infty  {\frac{{{n^\alpha }{q^n}}}{{1 - {q^n}}}}, \]
with $\alpha = -1$ and $q=e^{-2\pi}$. To obtain a different $\sigma_{\alpha}(n)$ term, with $\alpha \in -\mathbb{Z}_*$ we have to perform additional processes of integration by parts.
\subsection{Functional equation for the remainder}
By using the functional equation of the Riemann zeta function as well as the functional equation \eqref{functional_nonhom} of the non-holomorphic Einsenstein series we can isolate another functional equation for the remainder, namely
\begin{align} \label{FE_of_Q}
{\pi ^{1 - 2s}}\Gamma (s)Q (s,\tau ) = \Gamma (1 - s)Q (1 - s,\tau ).
\end{align}
We now ask the question of whether a functional equation can also be derived by contour integration. Let us first look at a simpler example.
In \cite{Kirsten3}, the operator
\[P  :=   - \frac{{{d^2}}}{{d{\tau ^2}}}\]
is considered. Under Dirichlet boundary conditions, the eigenvalues are given by
\[
\lambda_n = n^2,
\]
where $n$ is a non-negative integer. The spectral zeta function associated to this operator is
\begin{align} \label{sepctraldef}
{\zeta _P}(s) :=  \sum\limits_n {{{({\lambda _n})}^{ - s}}}  = \sum\limits_{n = 1}^\infty  {{n^{ - 2s}}}.
\end{align}
Let us for a moment pretend that we are not aware that this series represents the Riemann zeta function. The above series is convergent for $\real(s)> \tfrac{1}{2}$. With
\[F(\lambda )  :=  \frac{1}{{2i\sqrt \lambda  }}({e^{i\sqrt \lambda  \pi }} - {e^{ - i\sqrt \lambda  \pi }}),\]
the contour integral method yields
\[{\zeta _P}(s) = \frac{1}{{2\pi i}}\int_\gamma  {d\lambda {\lambda ^{ - s}}\frac{d}{{d\lambda }}\log F(\lambda )}  = \frac{{\sin (\pi s)}}{\pi }\int_0^\infty  { dx {x^{ - s}}\frac{d}{{dx}}\log \frac{1}{{2\sqrt x }}({e^{\sqrt x \pi }} - {e^{ - \sqrt x \pi }})}, \]
where $\gamma$ is the Hankel contour enclosing the eigenvalues $\lambda_n$ depicted in Fig.~2.1. To obtain the second integral, the deformation to the negative real axis depicted in Fig.~2.2 is performed. From the behaviour of the integrand, this integral representation is seen to be valid for $\tfrac{1}{2} < \operatorname{Re} (s) < 1$. Furthermore, splitting the integral as $\int_0^1 {dx}  + \int_1^\infty  {dx}$ yields
 \begin{align} \label{kl_main}
  {\zeta _P}(s) &= \frac{{\sin (\pi s)}}{{2s - 1}} - \frac{{\sin (\pi s)}}{{2s\pi }} + \frac{{\sin (\pi s)}}{\pi }\int_1^\infty  {dx{x^{ - s}}\frac{d}{{dx}}\log (1 - {e^{ - 2\sqrt x \pi }})}  \nonumber \\
   &+ \frac{{\sin (\pi s)}}{\pi }\int_0^1 {dx{x^{ - s}}\frac{d}{{dx}}\log\left[\frac{{{e^{\sqrt x \pi }}}}{{2\sqrt x }} (1 - {e^{ - 2\sqrt x \pi }})\right]}.
\end{align}
This is now valid for $-\infty < \real(s) < 1$ and it allows us to find $\zeta_P(0) = -\tfrac{1}{2}$ as well as $\zeta_P'(0) = -\log(2\pi)$, which was the goal in \cite{Kirsten3} (more specifically, the functional determinant of the operator $P$). In fact, the above expression also tells us that $\zeta_P(-k) =0$ for $k=1,2,3,\cdots$. To construct the analytic continuation over $\C$ we focus on $\real(s)<0$ so that, splitting the second integral as a sum of two logarithms, we obtain
\[{\zeta _P}(s) = \sin (\pi s)\int_0^\infty  {dx\frac{{{x^{ - s - 1/2}}}}{{{e^{2\pi \sqrt x }} - 1}}}. \]
Performing the change $s \to \tfrac{s}{2}$ and $y = 2\pi \sqrt x$, we can write the above as
\[{\zeta _P}(\tfrac{s}{2}) = {(2\pi )^{  s}}\frac{{\sin (\tfrac{{\pi s}}{2})}}{\pi }\int_0^\infty  {dy\frac{{{y^{ - s}}}}{{{e^y} - 1}}}. \]
Next, we set $s=1-u$ so that $\real(u) = 1 - \real(s) >1$, thus
\begin{align} \label{integral_FE}
{\zeta _P}(\tfrac{{1 - u}}{2}) = {(2\pi )^{  1 - u}}\frac{{\cos (\tfrac{{\pi u}}{2})}}{\pi }\int_0^\infty  {du\frac{{{y^{u - 1}}}}{{{e^y} - 1}}}.
\end{align}
By real variable methods \cite{Titchmarsh} this last integral, which represents the Mellin transform of $1/(e^y-1)$, is seen to yield
\[\int_0^\infty  {dy\frac{{{y^{u - 1}}}}{{{e^y} - 1}}}  = \Gamma (u){\zeta _R}({u}), \quad \real(u)>1.\]
Finally, by the use of the Euler reflection formula we have
\[{\zeta _P}(\tfrac{u}{2}) = {(2\pi )^{u - 1}}\Gamma (1 - u)\frac{{\sin (\pi u)}}{{\cos (\tfrac{{\pi u}}{2})}}{\zeta _P}(\tfrac{{1 - u}}{2}) = {2^u}{\pi ^{u - 1}}\Gamma (1 - u)\sin (\tfrac{{\pi u}}{2}){\zeta _P}(\tfrac{{1 - u}}{2}).\]
This is the functional equation of the spectral zeta function $\zeta_P(u)$. It shows that $\zeta_P(u)$ has a meromorphic continuation to $\C$ with a simple pole at $u=\tfrac{1}{2}$ with residue equal to $\tfrac{1}{2}$.\\
This is in fact a variation of the original techniques introduced by Riemann in order to obtain the functional equation and analytic continuation of his zeta function (see chapter 2 of \cite{Titchmarsh} for a detailed discussion).\\
The question is now whether a similar argument can be reproduced for the non-holomorphic Eisenstein series for which we know there exists a functional equation.\\
To make matters easier, we take the case $\tau = i$. Equation (4-5) of \cite{Kirsten4c} shows that the remainder $Q(s,i)$ can be written as
\[Q(s,i) = 4\frac{{\sin (\pi s)}}{\pi }\tau _2^s\sum\limits_{n = 1}^\infty  {\int_{{n^2}}^\infty  {dz{{(z - {n^2})}^{ - s}}\frac{d}{{dz}}\log (1 - {e^{ - 2\pi \sqrt z }})} } .\]
We can re-write this as
\begin{align} \label{remaindertaui}
Q(1 - s,i) = 4\frac{{\sin (\pi s)}}{\pi }\sum\limits_{n = 1}^\infty {\int_0^\infty du\,\, {{u^{s - 1}}\left( {\frac{1}{{{e^{2\pi \sqrt {{n^2} + u} }} - 1}}\,\,\frac{\pi }{{\sqrt {{n^2} + u} }}} \right)} },
\end{align}
so that in fact we are interested in computing the Mellin transform of
\[h(x) := \frac{1}{{{e^{2\pi \sqrt {{n^2} + u} }} - 1}}\,\,\frac{\pi }{{\sqrt {{n^2} + u} }}.\]
This is to be compared with \eqref{integral_FE} but in this case it is more difficult to evaluate the integral, except for the case $n=0$ for which we would go back to the Riemann zeta function. Thus, instead of obtaining the functional equation \eqref{FE_of_Q} we can use \eqref{FE_of_Q} to conclude that the infinite sum is
\[\sum\limits_{n = 1}^\infty  {\int_0^\infty  du\,\,{{u^{s - 1}}\left( {\frac{1}{{{e^{2\pi \sqrt {{n^2} + u} }} - 1}}\,\,\frac{\pi }{{\sqrt {{n^2} + u} }}} \right)} }  = \frac{{{\pi ^{ - 2s}}}}{4}Q(s,i).\]
An alternative way to produce the functional equation of the non-holomorphic Eisenstein series is as follows, see \cite{zagier}.
\begin{definition}
Let $\lambda_n$ be the eigenvalues of an operator $\Delta$ on a compact smooth manifold $M$ with Riemannian metric $g$. For $t>0$ the heat kernel is defined by
\[K(t) : =  \sum\limits_n {{e^{ - \lambda {}_nt}}}. \]
\end{definition}
\noindent Taking the eigenvalues to be $\left| {m + n\tau } \right|\pi \tau _2^{ - 1}$ we get the following result.
\begin{lemma}[Jacobi-inversion]
If
\[K(x,\tau ) := \sum\limits_{(m,n) \in {\Z^2}} {\exp ( - \left| {m + n\tau } \right|\pi \tau _2^{ - 1}x)} \]
with $x>0$ then the following functional equation holds
\begin{align} \label{jacobiinversion}
K(x,\tau ) = {x^{ - 1}}K({x^{ - 1}},\tau ).
\end{align}
\end{lemma}
\noindent This is proved by using Poisson summation in two variables with $f(x,y) = \exp ( - (u + v\tau )(u + v\bar \tau )\pi \tau _2^{ - 1}{x^{ - 1}})$. By making the change $t = {\left| {m + n\tau } \right|^2}\pi x \tau_2^{-1}$ in the integral representation of the $\Gamma$ function, summing over $(m,n) \in \Z_*^2$ and integrating term by term we see that for $\operatorname{Re}(s)>1$
\[\tilde E(s,\tau ) := \frac{1}{2}{\pi ^{ - s}}\Gamma (s){E^*}(s,\tau ) = \frac{1}{2}{\pi ^{ - s}}\Gamma (s)\sum\limits_{(m,n) \in \Z_*^2} {Q{{(m,n,\tau )}^{ - s}}}  = \frac{1}{2}\int_0^\infty  {(K (t,\tau ) - 1){t^{s - 1}}dt}, \]
where $Q(m,n,\tau ): = {\left| {m + n\tau } \right|^2}{\tau _2}$. Using \eqref{jacobiinversion} leads to
\[\tilde E(s,\tau ) = \tilde E(1 - s,\tau ).\]
\section{Conclusion and outlook}
We have shown in this paper that the method of contour integration applied to the eigenvalues of the Laplacian in the torus yields an alternative expression of the Chowla-Selberg series formula which is very useful to immediately and effortlessly obtain the functional determinant of the Laplacian operator as opposed to having to use sophisticated results such as the Kronecker first limit formula. In fact, once the equivalence with the Chowla-Selberg formula is established, the proof of the Kronecker limit formula follows. Moreover, as special cases of the comparison we obtain other interesting formulae, some of which are new and some already well-known, involving the Dedekind eta function as well as the divisor function.\\
One may, for instance, wish to apply this method to the case of algebraic fields and non-Euclidean space $\mathcal{H}^3$. The points of $\mathcal{H}^3$ are denoted by $z=(x,y)$ with $x=x_1 + x_2 i$ with $x_1,x_2 \in \R$ and $y>0$. The space $\mathcal{H}^3$ is embedded in the Hamiltonian algebra of quaternions so that $z = x + y \mathfrak{j} \in \mathcal{H}^3$ with $\mathfrak{j}^2=-1$ and $i \mathfrak{j} = - \mathfrak{j} i$. If we denote by $y(z)$ the third coordinate of $z$ and by $\mathbb{Z}[i]$ the Gaussian number field, then the Eisenstein series is defined by
\[E_{\Z[i]}^*(s,\tau )  :=  \frac{1}{4}\sum\limits_{\begin{subarray}{l}
  l,h \in \Z[i] \\
  (l,h) = 1
\end{subarray}}  {\frac{{{y^s}}}{{{{({{\left| {lx + h} \right|}^2} + {{\left| {ly} \right|}^2})}^s}}}} .\]
This series converges absolutely for $\operatorname{Re}(s)>2$ and its Fourier expansion is given by
\[E_{\Z[i]}^*(s,\tau ) = {y^s} + \frac{\pi }{{s - 1}}\frac{{{\zeta _K}(s - 1)}}{{{\zeta _K}(s)}}{y^{y - 2}} + \frac{{2{\pi ^s}y}}{{\Gamma (s){\zeta _K}(s)}}\sum\limits_{\begin{subarray}{l}
  n \in \Z[i] \\
  n \ne 0
\end{subarray}}  {{{\left| n \right|}^{s - 1}}{\sigma _{1 - s}}(n){K_{s - 1}}(2\pi \left| n \right|y){e^{2\pi i\operatorname{Re} (n\bar x)}}}, \]
where $\zeta_K(s)$ is the Dedekind zeta function of the Gaussian number field and $\sigma_\nu(n)=\tfrac{1}{4}\sum_{d|n}{\left| d \right|^{2\nu}}$ with $d \in \mathbb{Z}[i]$.
The techniques described here should allow for an alternative representation and derivation of this result.
\section*{Acknowledgments}
NR wishes to acknowledge partial support of SNF Grant No. 200020-131813/1 as well as Prof. Mike Cranston for the hospitality of the Department of Mathematics of the University of California, Irvine.
EE was supported in part by MICINN (Spain), grant PR2011-0128 and project FIS2010-15640, by the CPAN Consolider Ingenio Project, and by AGAUR (Generalitat de
Ca\-ta\-lu\-nya), contract 2009SGR-994, and his research was partly carried out while on leave at the Department of Physics and Astronomy, Dartmouth College, NH.

\textsc{Institute of Space Science, National Higher Research Council, ICE-CSIC, Facultat de Ciencies, Campus UAB, Torre C5-Par-2a,
08193 Bellaterra (Barcelona), Spain}\\
\href{mailto:elizalde@ieec.uab.es}{email: elizalde@ieec.uab.es; elizalde@math.mit.edu}\\

\textsc{GCAP-CASPER, Department of Mathematics, Baylor University, One Bear Place, Waco, TX 76798, USA}\\
\href{mailto:klaus_kirsten@baylor.edu}{email: klaus\_kirsten@baylor.edu}\\

\textsc{Institut f\"{u}r Mathematik, Universit\"{a}t Z\"{u}rich, Winterthurerstrasse 190, CH-8057 Z\"{u}rich, Switzerland}\\
\href{mailto:nicolas.robles@math.uzh.ch}{email: nicolas.robles@math.uzh.ch}\\

\textsc{Department of Mathematics and Statistics, Lederle Graduate Research Tower, 710 Noth Pleasant Street, University of Massachusetts, Amherst, MA 01003-9305, USA}\\
\href{mailto:williams@math.umass.edu}{email: williams@math.umass.edu}

\end{document}